\documentclass[envcountsame,runningheads,11pt]{llncs}
\usepackage[utf8]{inputenc}     
\usepackage[english]{babel}     
\usepackage[margin=1in,paper=letterpaper]{geometry}

\usepackage{graphicx}           
\usepackage{amsmath,amssymb}    
\usepackage{mathtools}          
\usepackage{enumitem}           
\usepackage{listings}           
\usepackage{todonotes}          
\usepackage{hyperref}           
\usepackage{mathrsfs}
\usepackage{xcolor}
\usepackage{xypic}
\usepackage{tikz}
\usepackage{float}
\usepackage{picinpar}
\usepackage{etoolbox}
\usetikzlibrary{quantikz2}

\AtEndEnvironment{proof}{\unskip\nobreak\hfill\raisebox{-0.2ex}{\ensuremath{\square}}\par}

\newcommand{\Supp}[1]{\text{Supp}_{#1}}
\newcommand{\outcome}{\mathcal{O}}
\newcommand{\G}{\mathcal{G}}

\newcommand{\Grec}{\G_{tr}}

\newcommand{\Gammarec}{{\Gamma_{tr}}}

\newcommand{\turn}[0]{\mathrm{turn}}

\definecolor{darkgreen}{rgb}{0,0.5,0}

\definecolor{darkred}{rgb}{.7,0,0}

\title{%
Quantum combinatorial games}
\subtitle{Can you beat a quantum opponent?}

\author{
	Dieks Scholten\inst{1}
	\and Bram Westerbaan\inst{2}
	\and Simona Samardjiska\inst{1}
} 
\institute{
   Radboud University, Nijmegen, Netherlands\\
    \email{dieks.scholten@ru.nl} \\
    \email{simonas@cs.ru.nl}
    \and
    \email{bram@westerbaan.name}
   }
\begin{document}
\maketitle
\begin{abstract}
A combinatorial game is a deterministic game with no hidden information played between two opponents such as 
tic-tac-toe, checkers or chess.
In this paper we extend combinatorial games to the quantum setting, by first 
 revisiting and reformulating existing theory of classical combinatorial games.

We investigate in which case a quantum opponent has an advantage over a classical one. Surprisingly, 
our instantiation of Zermelo's classical theorem in the quantum setting shows that the effects of quantum mechanics do not convey an advantage
against a classical 
player that plays a perfect classical strategy. In a more realistic scenario, when the classical player makes mistakes, 
we show how the quantum opponent can amplify the mistake to increase their chance of winning.
Our theory has application beyond the mere playing of board games and can be used as a tool in finite deterministic adversarial 
models with perfect information. 

\end{abstract}

\begingroup
\makeatletter
\def\@thefnmark{} \@footnotetext{
\footnotesize{%
This research has been supported by the Dutch government through the NWO Quantum Technology grant NGF.1623.23.020.
}}
\endgroup

\section{Introduction}\label{section:introduction}

A common characteristic of a game of tic-tac-toe, checkers, chess, abalone or go is that they are played between two opponents,
they are deterministic and don't involve any hidden information. 
In the literature, they are known as combinatorial games. 
In this paper we ask ourselves
what a match of a combinatorial game 
might look like 
on a quantum computer,
and whether one can have strategies that guarantee 
a certain outcome like winning,
or not losing (e.g.~tic-tac-toe).
While the rules for any of the above mentioned classical, well-known combinatorial games are set in stone,
different generalizations are possible for a quantum variant.
Indeed, several quantum variants of tic-tac-toe have already been proposed 
for different purposes~\cite{goff2006quantum,weingartner2023quantum,sagole2019quantum,nagy2012quantum,chiofalo2022games}, and also for chess~\cite{cantwell2019quantum} and go~\cite{ranchin2016quantum,sahu2022quantum,qiao2020quantum}.
None of them, however,
quite fit the following 
reasonable requirements:
\begin{enumerate}
\item 
The quantum variant should be actually practically playable on a quantum computer;
e.g.~measurement at the end.
\item 
The quantum variant should be an extension of the regular, classical game in the sense
that a strategy for the regular game (a set of instructions on how to react to each 
move of the opponent) should be playable in the quantum variant.
\item
The player in the quantum variant should not be arbitrarily restricted
to, say, make superpositions of only two moves.
\end{enumerate}
We propose a general framework for quantization of combinatorial games 
that meets all these realistic requirements and applies to a broad class of combinatorial games with ties,
including tic-tac-toe, chess and hex.
After having defined the exact rules of our quantum variant
of a combinatorial game,
we turn to the question of winning strategies.
For example, it is well known that there is no winning strategy for tic-tac-toe,
but there \emph{is} a non-losing one.
Our main result implies that,
surprisingly,
in our variant of quantum tic-tac-toe 
a `quantum player', a player that takes full advantage of the moves
permitted under this quantum variant,
can still not win against a regular `classical player' playing the classical non-losing
strategy.  A similar result holds for quantum combinatorial games in general, see Section~\ref{sec:quantum_comb_games}.

One might conclude from this that playing 
on a quantum computer (under our rules) is not really
interesting as the classical best strategies still apply.
This might well be the case for a game like tic-tac-toe,
where the classical best strategy is memorizable by a child,
but does not apply to games like chess, where, while
at least one of the two players (black or white)
must have a non-losing strategy (by Zermelo's theorem~\cite{zermelo1913}),
the details of such a strategy are not known to mortals.
It might thus not be unreasonable to think that a chess player
making proper quantum moves
might gain an edge over a non-perfect, human opponent playing just classical moves.
(We'll actually see an example of this, Lemma~\ref{amplify-mistake}).

Finally, taking into account the criteria above, we argue that the framework we propose can be used as a tool for many finite deterministic adversarial models with perfect information. One only has to find a suitable translation of their problem to a classical game. An example is the Byzantine Generals problem, which we detail in Section~\ref{sec:applications}.

\subsubsection*{Related Work.}
Several publications on quantized combinatorial games have been published in recent years. Especially \emph{quantum tic-tac-toe} has attracted considerable attention. The first definition was proposed by Goff et al.~\cite{goff2006quantum} who introduced quantum tic-tac-toe for teaching the concepts of quantum mechanics and computing. Subsequent definitions~\cite{weingartner2023quantum,sagole2019quantum,nagy2012quantum,chiofalo2022games} shared similar motivations. This justifies the use of a restricted set of quantum moves and sometimes involves intermediate measurement.
Building on top of Goff et al., Leaw \& Cheong \cite{leaw2010strategic} propose an alternative definition of quantum tic-tac-toe that generalizes classical tic-tac-toe (contrary to Goff et al). However, the implementability of this construction as a quantum circuit remains unclear.

Besides quantum tic-tac-toe, other combinatorial games have also been quantized. Cantwell's \cite{cantwell2019quantum} variant of quantum chess uses ``split'' and ``merge'' moves, which are a good way to convey intuition about quantumness but very restricting for the quantum player. Both Ranchin et al. and Sahu et al.~\cite{ranchin2016quantum,sahu2022quantum} construct a quantum-like version of go for pedagogical purposes. Qiao et al.'s \cite{qiao2020quantum} variant of quantum go expands classical go by adding specific quantized moves. Finally, a variant of quantum checkers has been proposed by Raat et al.\ \cite{raat2025quantum}, implementing the proposal for general combinatorial games by Dorbec and Mhalla~\cite{dorbec2017toward}.

Dorbec and Mhalla~\cite{dorbec2017toward} propose a general framework for quantum combinatorial games. They propose a definition with five levels of quantumness. The level with the most ``quantumness'', allows both classical moves and moves in a superposition all with equal amplitude. The moves must be applied to all positions in the superposition, illegal moves are dropped. A player loses a position once measuring a loss has probability one. This was done to ensure the combinatorial nature of the quantum variant of the games.

Kyle Burke et al.~\cite{burke2020quantum} study several properties of the various quantumness levels in the framework of Dorbec and Mhalla, most notably (for our research) showing that quantum moves \emph{have} impact on the game's outcome, in their highest quantumness level. This is contrary to our results and the consequence of Dorbec and Mhalla not adhering to our third requirement:  the players' moves in the quantum game should not be arbitrarily restricted. Specifically, in the game defined by Dorbec and Mhalla a player cannot make the move dependent on the specific position in the superposition it is applied to --- something which should be achievable using controlled gates --- possibly forcing players to make mistakes.


\subsubsection*{Our contribution}
We propose a general framework for quantizing combinatorial games. First, we show that any combinatorial game can be reformulated into an equivalent \emph{quantizable} form --- one with a fixed number of rounds and injective valid moves --- which we establish via bisimulation.

Our main result is that quantumness does not improve a player's achievable outcome against a perfect opponent. Formally, we show that for any position the set of guaranteeable outcomes is identical (up to superpositions) in the classical and quantum variants of the game. The key insight is that quantum moves must respect classical rules, so if an outcome can be guaranteed classically the quantum moves cannot escape this as their effects can be reduced to a set of classical moves. 

As a corollary, we prove a quantum Zermelo theorem~\cite{zermelo1913}. Applying the classical theorem directly to the quantum game --- which is itself a combinatorial game --- yields guarantees only over \emph{distributions} of outcomes. In particular, when neither player can force an outright win, it leaves open the possibility of losing with some probability. Our stronger version eliminates this: if neither player can force a win, both can still prevent any chance of losing, just as in the classical game.
Against an imperfect opponent, we show that quantum players \emph{can} gain advantage. To do so, we introduce a \emph{flaw-amplifying} strategy that ``cherry-picks'' the opponent's mistakes, yielding a strictly better outcome than classical strategies could achieve.

To demonstrate the applicability of our framework beyond merely playing games, we formulate the Byzantine Generals problem as a combinatorial game. This yields a clean proof that any instance of the Byzantine Generals problem solvable under classical failures is also solvable under quantum failures --- for classical well-behaving nodes. 

\subsubsection*{Organization of the paper.}
The paper is organized as follows. In Section~\ref{sec:comb_games} we define combinatorial games 
and develop the basic combinatorial theory.
We extend our theory to quantum combinatorial games in Section~\ref{sec:quantum_comb_games}, where we analyse the cases of 
a perfect and an imperfect classical players. In Section~\ref{sec:applications}, we present a non-trivial application 
of our theory. We conclude with Section~\ref{sec:conclusion} where we discuss possible generalizations and open problems.

\section{Formal approach to combinatorial games}\label{sec:comb_games}

We will be interested in two-player games like tic-tac-toe,
where no chance or secrets are involved,
and that are guaranteed to end after a set number of turns
with a win for one of the two players, or a tie.
We will denote the two players with $X$ and $O$.
More formally we define such games as follows. 
Note that this is not the standard definition
of combinatorial game\cite{WinningWays},
but a version modified for our purposes.

One change worth mentioning up front is that 
we include whose turn it is into the positions of our combinatorial games.
So while normally a (combinatorial game) position of chess would only comprise
the location of the pieces on the chess board\footnote{And, of course, whether black has castled, etc.},
our chess positions include whether black or white is to move too.
\begin{definition}\label{cg}
	A \textbf{combinatorial game} $\mathcal{G}$
	consists of a set of \textbf{positions} (also denoted by) $\mathcal{G}$,
	a set of \textbf{moves} $\Sigma$,
	a \textbf{starting position} $0$,
	a \textbf{player-to-move map} $\turn\colon \mathcal{G}\to \{X,O\}$,
	a \textbf{transition map},
	$$\Gamma\colon \mathcal{G} \times \Sigma \longrightarrow \mathcal{G}
	\cup \{ \perp \},$$
	a set of \textbf{finished positions} $\mathcal{F}\subseteq \mathcal{G}$,
	an \textbf{outcome set} $\Omega$ (usually $\{X,O,T\}$), and
	an \textbf{outcome function} $\mathcal{O}\colon \mathcal{F}\rightarrow \Omega$,
	such that
	(we write $g\xrightarrow{\sigma} g'$ when $\Gamma(g,\sigma)=g'$):
	\begin{enumerate}
		\item ($X$ starts) $\turn(0)=X$;
		\item (Taking turns) $\turn(\Gamma(g,\sigma)) \neq \turn(g)$ for all $g\in \mathcal{G}$
			and $\sigma\in \Sigma$ with $\Gamma(g,\sigma)\neq\perp$;
		\item\label{no-moves-implies-finished}(No moves implies finished)
			If in a position~$g$ there is no more (valid) move,
			then  $g\in \mathcal{F}$;
		\item\label{outcomes-do-not-change} (Outcomes do not change)
			If in a finished position $g\in \mathcal{F}$
			it is still possible to make a move,
			$g\xrightarrow{\sigma} g'$,
			then this does not change the outcome:
			$g'\in\mathcal{F}$ and $\mathcal{O}(g')=\mathcal{O}(g)$.
		\item\label{reachable-positions} (Each position is reachable) For every $g\in \mathcal{G}$
			there are $\sigma_1,\dotsc,\sigma_n\in \Sigma$  such that
			$$0\xrightarrow{\sigma_1} g_1 \xrightarrow{\sigma_2} g_2 
			\xrightarrow{\sigma_3} \dotsb 
			\xrightarrow{\sigma_n} g_n=g.$$
		\item\label{games-end} (Games end)
		There is a natural number~$N$
			such that for 
			any sequence $g_0\xrightarrow{\sigma_1}g_1 \xrightarrow{\sigma_2} \dotsb \xrightarrow{\sigma_N}g_N$
			of~$N$ moves we have  $g_N\in\mathcal{F}$.
	\end{enumerate}
We call a combinatorial game \textbf{finite}
if the following three sets are finite:
the set of positions, $\mathcal{G}$, the set of moves, $\Sigma$, and the outcome set, $\Omega$.
\end{definition}
Some remarks about the definition above are in order.

Only for the \emph{finite} combinatorial games
will we define a quantum variant,
which will itself be a (non-finite) combinatorial game.
The restriction to a finite set of positions and moves
is necessary,
because (quantum) computers only have finite memory.
The resulting quantum variant, however,
will \emph{not} have only a finite set of positions,
so finiteness is not part of the definition of combinatorial game above.
Similarly, since the execution time on a quantum computer is limited,
we only consider games that end after finitely many moves.
(These are known as `short' games in the literature~\cite{WinningWays}.)
Interestingly,
the quantum variant will be short as well,
so we included this restriction in the definition of our `combinatorial game',
via condition~\ref{games-end}.
While the outcome set for tic-tac-toe 
(and many other combinatorial games)
is simply $\{X,O,T\}$, representing a win for~$X$, a win for~$O$, and a tie, respectively,
you will see below
that the outcome set for a game of quantum tic-tac-toe
will be the set~$\mathcal{D}(\{X,O,T\})$ of probability distributions over~$\{X,O,T\}$.

Condition~\ref{outcomes-do-not-change} might seem odd:
why not simply prohibit moves in a finished position?
The reason is that a quantum position will be a superposition of classical positions,
some of which may already be finished while others are not.
If we want a move on a quantum position to simply be a classical move
on each component of the superposition,
we must permit moves on finished positions.
In many games, such as tic-tac-toe, this is not part of the original rules,
so for quantization we consider a slightly modified version
with a `skip' move on any finished position
(see Definition~\ref{def: extended game}),
which raises the question whether anything essential has changed.
We address this by showing the modified game is \emph{bisimilar}
to the original (Definition~\ref{def:bisimulation}).

A second hurdle one encounters when defining a quantum variant
of a combinatorial game
is making sure that a classical strategy
can be followed on a quantum computer as well.
It may happen that two different positions
lead to the same position in one step,
for example, for tic-tac-toe it is standard
for~$X$ to move to ii)
both from i) and iii) in the figure below.
This presents a problem on
\newsavebox{\injDiagram}%
\sbox{\injDiagram}{\includegraphics[width=0.4\textwidth]{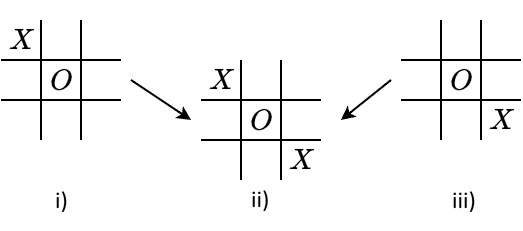}}%
\begin{window}[0,r,{\usebox{\injDiagram}},{}]
\noindent a quantum computer: any operation
must be a unitary, which is injective,
which the move iii) to ii) is not.
The solution is to let the unitary act not on the current position alone,
but on the full play history (Definition~\ref{def:transcribed games}),
making $\Gamma\colon \mathcal{G}\times \Sigma\to \mathcal{G}\cup\{\perp\}$ an injective partial map.
Again, bisimulation shows this modified game is equivalent, at least classically.
\end{window}

\subsection{Basic combinatorial game theory}

\begin{definition}
Given a position~$g$ in a combinatorial game~$\mathcal{G}$,
let $N(g)$ be the maximum number of moves
that can be made from~$g$ before ending up in a finished position.
\end{definition}
Note that~$N(g)<\infty$ by condition~\ref{games-end}.
The number~$N(g)$
is very convenient for induction
over the set of positions in a game, because:
\begin{enumerate}
\item 
	$N(g)=0$ if and only if~$g$ is finished; and
\item
	if $g\xrightarrow{\sigma} g'$ for some~$\sigma$, and~$g$ is not finished, then $N(g') < N(g)$.
\end{enumerate}
(The number  $N(g)$ plays a role similar to the \emph{birthday} of a position in traditional
combinatorial game theory.)

With the help of~$N(g)$
we can introduce
the two most important characters
of this paper, the sets~$\Sigma_{\mathcal{T}}$ and $\Pi_{\mathcal{T}}$
(where $\mathcal{T}\subseteq \Omega$ is some `target outcome'),
with the following interpretation:
\begin{itemize}
	\item $g\in \Sigma_{\mathcal{T}}$
		iff the current player can achieve an outcome
		from~$\mathcal{T}$, whatever the opponent does.
	\item $g\in \Pi_{\mathcal{T}}$
		iff no matter what the current player does,
		the opponent can achieve an outcome from~$\mathcal{T}$.
\end{itemize}
\begin{definition}\label{def:Sigma_T Pi_T}
	For a combinatorial game~$\mathcal{G}$
	and `target outcome'~$\mathcal{T} \subseteq \Omega$ we define
	$\Sigma_{\mathcal{T}}, \Pi_{\mathcal{T}}\subseteq \mathcal{G}$ by:
	\begin{itemize}
	\item $g\in \Sigma_{\mathcal{T}}$
		iff either~$g$ is finished and  $\mathcal{O}(g)\in\mathcal{T}$,
			or $g$ is not finished and
			there is some move $g\xrightarrow{\sigma} g'$
			with $g'\in \Pi_{\mathcal{T}}$.
	\item $g\in \Pi_{\mathcal{T}}$
		iff either~$g$ is finished and  $\mathcal{O}(g)\in\mathcal{T}$, 
			or $g$ is not finished and  $g'\in \Sigma_{\mathcal{T}}$ for 
			any move  $g\xrightarrow{\sigma} g'$.
	\end{itemize}
\end{definition}
\begin{remark}
It is not a priori clear that this definition
is well-founded;
if we completely unfold the recursive definition
of $g\in \Sigma_{\mathcal{T}}$,
why might we not  end up with an infinite recursion tree
leaving us unable to decide whether $g\in \Sigma_{\mathcal{T}}$ or not?
This problem, however, does not occur,
since games do not loop nor go on forever. The recursion tree will have depth at most~$N(g)$,
and so we end up always being able to decide
whether $g\in \Sigma_{\mathcal{T}}$
and whether  $g\in \Pi_{\mathcal{T}}$.
\end{remark}

Note that the interpretation
for 
$g \notin \Sigma_{\mathcal{T}}$
is that the current player cannot make sure that the game ends in~$\mathcal{T}$
if their opponent makes the right moves.
Logically, this should
be equivalent
to their opponent having a way to make sure the game ends in an outcome from
$\Omega\backslash \mathcal{T}$,
whatever the current player does, that is,  $g\in \Pi_{\Omega\backslash \mathcal{T}}$.
This is indeed the case.
\begin{theorem}\label{thm:fund-cgt}
Given a combinatorial game~$\mathcal{G}$, and~$\mathcal{T} \subseteq \Omega$,
	we have $\mathcal{G}\backslash \Sigma_{\mathcal{T}} 
	= \Pi_{\Omega\backslash\mathcal{T}}$.
\end{theorem}
\begin{proof}
Let us write $\overline{\mathcal{T}} := 
\Omega\backslash \mathcal{T}$ for brevity.
	Let~$g\in \mathcal{G}$.
We'll prove the following equation holds for all~$\mathcal{T}$ by
strong induction over~$N(g)$.
	\vspace{-1.2em}

\begin{equation}
\label{pisigma}
g\notin\Sigma_{\mathcal{T}} \ \iff \ \ g\in\Pi_{\Omega\backslash \mathcal{T}}.
\end{equation}
	When~$N(g)=0$, then~$g\in \mathcal{F}$,
	and so $g\notin \Sigma_{\mathcal{T}}$ iff not $\mathcal{O}(g) \in \mathcal{T}$
	iff $\mathcal{O}(g) \in \overline{\mathcal{T}}$
	iff $g\in \Pi_{\overline{\mathcal{T}}}$, for any~$\mathcal{T}$.

Suppose now that~$N(g)>0$, and that~\eqref{pisigma} holds for all~$g'$ with~$N(g')<N(g)$
	and all $\mathcal{T}$;
	we'll show that~\eqref{pisigma} holds for~$g$ and all~$\mathcal{T}$ as well.
Note that~$g$ is not finished,
	since $N(g)>0$.

	Given~$\mathcal{T}$, the following are equivalent to $g\notin \Sigma_{\mathcal{T}}$:
\begin{itemize}
\item there are no $g'$ and~$\sigma$ with  $g\xrightarrow{\sigma}g'$
	and  $g'\in \Pi_{\mathcal{T}}$;
\item for any $g'$ and $\sigma$ with $g\xrightarrow{\sigma}g'$ we have  $g'\notin \Pi_{\mathcal{T}}$; and 
	(using the induction hypothesis here, and the fact that $N(g')<N(g)$)
\item for any $g'$ and $\sigma$ with $g\xrightarrow{\sigma}g'$ we have  $g'\in \Sigma_{\overline{\mathcal{T}}}$.
\end{itemize}
Since the latter item is equivalent to $g\in \Pi_{\overline{\mathcal{T}}}$, we have proven~\eqref{pisigma} by induction.
\end{proof}

\begin{corollary}[Zermelo]\label{cor: zermelo's result}
Given a combinatorial game~$\G$ and a partition
of the outcome set
	$\Omega = \Omega_X \sqcup \Omega_O \sqcup \Omega_T$,
	where an outcome in $\Omega_X$ is interpreted as a win for~$X$,
	$\Omega_O$ as a win for~$O$,
	and $\Omega_T$ a tie.
	Then for a position~$g$ of~$\G$ where the current player is~$X$
	exactly one of the following three mutually exclusive cases holds:
	\begin{enumerate}
		\item\label{zermelo-X} $g \in \Sigma_{\Omega_X}$ --- the current player~$X$ can force a win;
		\item\label{zermelo-O} $g \in \Pi_{\Omega_O}$ --- their opponent~$O$ can force a win;
		\item\label{zermelo-T} $g \in \Sigma_{\Omega_T \sqcup \Omega_{X}}$ and $g \in \Pi_{\Omega_T \sqcup \Omega_{O}}$ 
			--- both players can prevent a loss.
	\end{enumerate}
(The same result holds
for when~$O$ is the current player,
	mutatis mutandis.)
\end{corollary}
\begin{proof}[See Appendix~\ref{omitted-proofs}]
\end{proof}
\begin{remark}
Note that if there are no ties, that is  $\Omega_T=\varnothing$,
only cases 1 and 2 can hold.
\end{remark}
\begin{remark}
	A special case of this corollary is where $g = 0$. We can therefore classify games on the case in which their starting position falls into. For example, tic-tac-toe falls in case~3:  neither the starting player~$X$ nor its opponent~$O$ can always guarantee a win, which means that they can both prevent a loss. 
\end{remark}

Before we can move to quantization of a combinatorial game we require two properties of a combinatorial game.
First, we need some notion of injectivity, as described above. We now define \textit{injective} combinatorial games and later in Lemma~\ref{lem: all classical moves have an equivalent quantum move} we will show that this requirement is sufficient
to ensure any classical move has some equivalent unitary quantum move.
\begin{definition}\label{def:injective combinatorial games}
We call a combinatorial game $\G$ \textbf{injective}
	when the transition map $\Gamma\colon \G\times \Sigma\rightarrow \G \cup\{\perp\}$
is injective on valid moves, that is,
	$\Gamma(g,\sigma)=\Gamma(h,\pi)\neq \perp$
	implies $g=h$ and~$\sigma=\pi$.
\end{definition}

To define a quantum variant,
we need not only
injectivity of the game,
and the fact that the game ends after a set number of moves (condition~\ref{games-end} of~Def.~\ref{cg}),
but also that that number of moves can always be played:
\begin{definition}\label{def:fixed-length combinatorial games}
We call a combinatorial game $\G$ \textbf{fixed-length} when 
there is a natural number~$N$
such that not only 
for any sequence
$g_0 \xrightarrow{\sigma_1} \cdots \xrightarrow{\sigma_N} g_N$ of $N$ moves we have $g_N \in \mathcal{F}$,
	but also 
	that for any sequence 
	$0 =: g_0 \xrightarrow{\sigma_1} \cdots \xrightarrow{\sigma_{M}} g_{M}$ of $M<N$ moves 
	there is still some move $\smash{g_{M}\xrightarrow{\sigma_{M+1}} g_{M+1}}$
	that can be played
	(even when~$g_M$ is already finished).
\end{definition}

Any combinatorial game
can be modified to become injective and of fixed-length,
without fundamentally changing the gameplay.
This follows from the fact that
the modified game is bisimilar (see Definition~\ref{def:bisimulation}) to the original one
(see Corollary~\ref{modified-game-bisimilar} and Proposition~\ref{prop:guaranteeable sets are equivalent for bisimilar games}).
Details of this construction can be found in Appendix~\ref{injective-and-fixed-length},
but we will introduce the notion of bisimilarity here.

The idea of a bisimulation
between two games $\mathcal{G}_1$ and $\mathcal{G}_2$
is that it allows one to translate a strategy for one game
into another. Every move in~$\mathcal{G}_1$ (either from the starting player or
their opponent)
can be translated into an equivalent move in~$\mathcal{G}_2$,
and vice versa.
\begin{definition}\label{def:bisimulation}
A \textbf{bisimulation}
$\sim$ between combinatorial games $\mathcal{G}_1$ and $\mathcal{G}_2$
with the same outcome set~$\Omega$
	is a relation $\sim \subseteq \mathcal{G}_1\times\mathcal{G}_2$
	such that $0_{\mathcal{G}_1} \sim 0_{\mathcal{G}_2}$, and:
\begin{enumerate}
\item
Given~$g_1 \sim g_2$ we have $g_1\in \mathcal{F}$ iff~$g_2\in \mathcal{F}$,
		and then $\mathcal{O}(g_1)=\mathcal{O}(g_2)$;
\item
Given $g_1 \sim g_2$ with $g_1 \notin \mathcal{F}_1$ (iff $g_2 \notin \mathcal{F}$),
we have:
	\begin{enumerate}
	\item If $g_1\xrightarrow{\sigma_1} g_1'$,
		then $g_2\xrightarrow{\sigma_2} g_2'$
		for some~$\sigma_2$ and~$g_2'$ with $g_1'\sim g_2'$.
	\item If $g_2\xrightarrow{\sigma_2} g_2'$,
		then $g_1\xrightarrow{\sigma_1} g_1'$
		for some~$\sigma_1$ and~$g_1'$ with $g_1'\sim g_2'$.
	\end{enumerate}
\end{enumerate}
When such a bisimulation exists,
	we call the games $\G_1$ and~$\G_2$
	\textbf{bisimilar},
	and we write $\G_1 \sim \G_2$.
\end{definition}

Bisimilar games have the same gameplay (for proof see Appendix~\ref{omitted-proofs}):
\begin{proposition}\label{prop:guaranteeable sets are equivalent for bisimilar games}
Let~$\sim$ be a bisimulation between games $\mathcal{G}_1$ and~$\mathcal{G}_2$
with the same outcome set~$\Omega$,
and let~$\mathcal{T}\subseteq \Omega$
	be some target outcome.
Then:
	\begin{equation}
		\label{bisimilarity-implies-equistrategy}
		g_1 \sim g_2 \quad \implies \begin{cases}
			\quad g_1 \in \Sigma_{\mathcal{T}} & \iff \quad
	g_2 \in \Sigma_{\mathcal{T}}\\ 
			\quad g_1 \in \Pi_{\mathcal{T}} & \iff \quad
	g_2 \in \Pi_{\mathcal{T}}
		\end{cases}.
	\end{equation}
\end{proposition}

\section{Quantum combinatorial games}\label{sec:quantum_comb_games}
Below we'll define a quantum variant $\mathcal{G}^Q$ of a combinatorial game~$\mathcal{G}$.
In this definition we've very explicitly distinguished
the components of the quantum and classical
games using the~$Q$ and~$C$ superscripts
(e.g.~$\Sigma^Q$ and $\Sigma^C$).
In the rest of the text,
we'll often drop the~$C$ superscript (just writing $\Sigma$ for~$\Sigma^C$).
\begin{definition}\label{qv}
	Let~$\mathcal{G}$ be a fixed-length and injective combinatorial game
	such that the positions can be encoded using length~$L$ bitstrings, 
	so~$\G\subseteq \{0,1\}^L$.
	We call such combinatorial games \textbf{quantizable}.
	The \textbf{quantum variant} $\mathcal{G}^Q$ of~$\mathcal{G}=:\mathcal{G}^C$ 
	is the combinatorial game with:
	\begin{enumerate}
			\item Set $\Sigma^Q$ of possible \textbf{moves}, defined as the set of all $2^L\times 2^L$ unitaries $U$;
		\item Set $\mathcal{G}^Q$ of possible  \textbf{positions} defined as the superpositions 
$$\textstyle \ket\varphi\,= \,\sum_{g\in \mathcal{G}^C} \alpha_g \ket{g}
			\qquad\text{so}\qquad
\sum_{g\in \mathcal{G}^C} \left|\alpha_g\right|^2=1$$
			that are reachable from the \textbf{starting position} $\ket{0}$ using the moves from $\Sigma^Q$.

Given such a position $\ket\varphi$,
we denote by $\Supp{}{\ket\varphi}$
the \textbf{support} of~$\ket\varphi$,
the set of all classical positions~$g\in\G^C$ with $\alpha_g \neq 0$;

\item Set $\mathcal{F}^Q$ of \textbf{finished positions} defined as the positions $\ket\varphi$ with
$\Supp{}(\ket{\varphi})\subseteq \mathcal{F}^C$;

\item \textbf{Transition map} $\Gamma: \mathcal{G}^Q \times \Sigma^Q \rightarrow \mathcal{G}^Q$ defined by

	$$\Gamma(\ket{\varphi}, U) =
	\begin{cases}
		{U \ket{\varphi}}, \text{whenever $U$ is a valid move on position}~ \ket\varphi \\
		=\perp, \text{otherwise}.
	\end{cases}$$
We say that a move~$U$ is \textbf{valid}\footnote{This definition of `valid' quantum move is very permissive
to the quantum player. This is on purpose: the weaker the condition on the quantum player,
			the stronger the main result (Theorem~\ref{no-quantum-advantage}) becomes.  A more restrictive, and perhaps more natural
			notion of validity is discussed in Remark~\ref{completely-valid}.} on a position~$\ket\varphi$
if $\ket\varphi$ is not finished
and $U\ket{\varphi}$ is a superposition of positions in~$\G^C$
and for all $g'\in \Supp{}(U\ket{\varphi})$ we have $g\xrightarrow{\sigma} g'$ for some $\sigma\in \Sigma^C$ and $g \in \Supp{}(\ket\varphi)$;
\item \textbf{Outcome set} $\Omega^Q$ defined as the set of probability distributions on classical 
	outcomes~$\Omega^C$, i.e.,
			$$\Omega^Q\  :=\ \mathcal{D}(\Omega^C)\ =\ \textstyle \{\  \mu \colon \Omega^C\to [0,1]\ \text{such that} \ \sum_{\omega\in \Omega^C} \mu(\omega)=1\  \};$$
\item \textbf{Outcome function} $\mathcal{O}:\mathcal{F}^Q\rightarrow \Omega^Q$ given by 
$\mathcal{O}\bigl(\,\sum_{g\in \mathcal{G}^C} \alpha_g \ket{g}\ \bigr)(\omega) \ = \
			\sum_{\substack{g\in \mathcal{G}^C \\ \mathcal{O}(g)=\omega}} \left|\alpha_g\right|^2;
$
\item \textbf{Player-to-move map} given by $\mathrm{turn}(\ket{\varphi}) = \mathrm{turn}(g)$ for any (and all) $g\in \Supp{}(\ket{\varphi})$. 
	(See Remark~\ref{cg-turn-wd} below for why this is well-defined.)
	\end{enumerate}
\end{definition}
\begin{remark}\label{cg-turn-wd}
That~$\turn$ is well-defined in the sense
that
$\turn(g_1)=\turn(g_2)$ for any~$g_1,g_2\in \Supp{}{\ket{\varphi}}$
follows from the fact that~$\ket\varphi$ is reachable by a sequence of, say, $M$ valid (quantum)
moves starting at~$\ket0$.
This implies that~$g_1$ (and also $g_2$) can be obtained
using a sequence of~$M$ classical moves starting at~$0$.
			If~$M$ is even, $\turn(g_1)=X=\turn(g_2)$,
			and if~$M$ is odd, $\turn(g_1)=O=\turn(g_2)$.
\end{remark}
\begin{remark}
The fact that $\G^Q$ obeys the six conditions laid out for combinatorial games
in Definition~\ref{cg} is pretty straightforward, except
for condition~\ref{no-moves-implies-finished},
which demands that
if
a (super)position~$\ket\varphi$ of~$\G^Q$
is not finished,
then there must still be some valid move~$U$ on it.
To see why this is the case,
first note that there must be a component~$g\in \Supp{}(\ket\varphi)$
of~$\ket\varphi$ that is not finished, $g\notin \mathcal{F}$.
Since the original game~$\G$ obeys condition~\ref{no-moves-implies-finished},
there must be a valid move on~$g$ in~$\G$.
Since~$\G$ is of fixed-length,
there is, in fact, still a valid move
	$\smash{g'\xrightarrow{\omega(g')}\dotsc}$
	on every $g'\in \Supp{}(\ket\varphi)$ and not just on~$g$.
Now we will see in Lemma~\ref{lem: all classical moves have an equivalent quantum move} below
that this map $\omega\colon \Supp{}(\ket\varphi)\rightarrow \Sigma$
	can be turned into a unitary~$U$
	that is a valid move on~$\ket\varphi$.
Hence~$\G^Q$ obeys condition~\ref{no-moves-implies-finished}.
\end{remark}
\begin{remark}
If~$\G$ is not injective, then we might still try to define~$\G^Q$,
but the resulting game~$\G^Q$ 
might not obey condition~\ref{no-moves-implies-finished}.
Indeed,
consider the game~$\G$ with four states $0:=00$, $1:=01$, $2:=10$ and~$F:=11$
from~$\{0,1\}^2$,
where~$0$ is the starting state, and three moves $\{1,2,\bullet\}$,
with the following transitions (and no other valid moves),
as depicted below, on the right. Here, $\turn(0)=\turn(F)=X$, and $\turn(1)=\turn(2)=O$.
The only finished state in this game is~$F$. For this example, the outcome set and function do not really matter,
so let's say that $\Omega=\{42\}$ and~$\mathcal{O}(F)=42$.
This is clearly a non-injective combinatorial game
of fixed-length, $N=2$.

In the quantum variant of this game
we can reach the state $\ket{\varphi} = \frac{1}{\sqrt{2}}(\ket{1} + \ket{2})$
from the starting position~$\ket{0}$
with the unitary~$U$ given by the circuit on the right.
This $U$ is a valid move in~$\G^Q$ on~$\ket{0}$,
because~$\ket\varphi$ is a superposition
of states~$1$ and~$2$,
which can both be reached from~$0$ in~$\G$ in one step.
\newsavebox{\noninjGraph}%
\sbox{\noninjGraph}{$\xymatrix@C=1.5em{
	& 0\ar[ld]_1\ar[rd]^2 & \\
	1\ar[rd]_{\bullet}& & 2\ar[ld]^{\bullet} \\
& F &
}$}%
\newsavebox{\noninjCircuit}%
\begin{lrbox}{\noninjCircuit}%
\begin{quantikz}
	&	\gate{H} & \ctrl{1}  &          &\\
	&		 & \targ{}   & \targ{}  &
\end{quantikz}%
\end{lrbox}%
\newsavebox{\noninjBoth}%
\sbox{\noninjBoth}{%
	\begin{minipage}{\wd\noninjCircuit}
		\centering
		\usebox{\noninjGraph}\\[1ex]
		\usebox{\noninjCircuit}
	\end{minipage}%
}%
\begin{window}[0,r,{\usebox{\noninjBoth}},{}]
Note that this state~$\ket\varphi$ is not finished in~$\G^Q$,
because~$1$ and~$2$ are not finished in~$\G$.
However, there is no valid move~$U$ on~$\ket\varphi$.
Indeed,
what can~$U\ket{1}$ be?
For~$U$ to be a valid move on $\ket\varphi$, 
$U\ket{1}$ must be a superposition of~$\ket{F}$,
since~$1\xrightarrow{\bullet}F$ is the only valid move in~$\G$.
So $U\ket{1}=\alpha_1\ket{F}$ for some~$\alpha_1\in\mathbb{C}$ with $\left|\alpha_1\right|=1$.
Similarly, $U\ket{2} = \alpha_2\ket{F}$
for some $\alpha_2\in\mathbb{C}$ with $\left|\alpha_2\right|=1$.
But then $U(\alpha_1^{-1}\ket{1})=\ket{F}=U(\alpha_2^{-1}\ket{2})$,
which contradicts the injectivity of the unitary~$U$.
Hence no such valid move~$U$ exists.
Thus~$\G^Q$ is not a combinatorial game, for this non-injective~$\G$.
\end{window}
\end{remark}

To show that $\G^Q$ is an extension of~$\G^C$, we need the following.
\begin{definition}
\label{classical-quantum-move}
A move $U\in \Sigma^Q$  on a position $\ket{\varphi}$ 
in the quantum variant $\mathcal{G}^Q$
of a combinatorial game 
is called \textbf{classical}
if for every~$g\in \Supp{}(\ket{\varphi})$
there exists $g' \in \mathcal{G}$ with $U\ket{g}=\ket{g'}$.
\end{definition}


\begin{lemma}\label{lem: all classical moves have an equivalent quantum move}
Let $\mathcal{G}$ be a quantizable combinatorial game. 
Let~$S\subseteq \G$ be a subset of positions,
and let~$\omega\colon S\to \Sigma$
be a map such that~$\omega(g)$ is a valid move on~$g$
for each~$g\in S$.
One may think of~$\omega\colon S\to\Sigma$
as a (partial) strategy for a classical player:  when faced with~$g\in S$,
the move~$\omega(g)$ is to be made.
There is a unitary~$U_\omega$
that implements this strategy,
that is,
\begin{equation}
\label{eq:cqm}
U_\omega\ket{g} \ =\  \ket{\Gamma(g,\omega(g))}\qquad\text{for all }g\in S.
\end{equation}
Moreover, $U_\omega$ is a valid and classical move
	on every non-finished $\ket\varphi\in \G^Q$ with $\Supp{}\ket\varphi \subseteq S$.
\end{lemma}
\begin{proof}[See Appendix~\ref{omitted-proofs-quantum}]
\end{proof}

\begin{definition}
Let $\mathcal{G}$ be a quantizable combinatorial game.
We denote by $\mathcal{G}^{C\text{--}Q}$
the \textbf{classical against quantum} variant
of~$\mathcal{G}^Q$
in which the starting player, $X$, is restricted to classical moves
while its opponent, $O$, is not.
Note that this restriction in moves of~$X$
makes some positions
of the general quantum variant $\mathcal{G}^{Q}$
unreachable in~$\mathcal{G}^{C\text{--}Q}$.
	(For example,
under $\mathcal{G}^{C\text{--}Q}$
the starting player $X$ 
can not move from~$0$ to a proper superposition,
while under
$\mathcal{G}^{Q}$ they can.)

The \textbf{quantum against classical}, $\mathcal{G}^{Q\text{--}C}$ variant
 (where instead $O$ is restricted to classical moves),
and the \textbf{classical against classical}, $\mathcal{G}^{C\text{--}C}$
variant
	(where both~$X$ and~$O$ are restricted to classical moves)
are defined similarly.
All three variants are 
a special case of the \textbf{partially quantum} variant, $\mathcal{G}^{Q|P}$,
	of $\mathcal{G}^{Q}$
	in 
	which both players
	may make quantum moves
	on the (super)positions
	from some set~$P\subseteq \mathcal{G}^{Q}$,
	but are restricted to classical moves
	on $\mathcal{G}^Q\backslash P$.
\end{definition}

With this conversion from classical  to classical quantum moves we show that classical games are equivalent to their quantum variant restricted to classical moves 
(for proof see Appendix~\ref{omitted-proofs}).

\begin{proposition}\label{classical-bisimilar}\label{prop: G bisimilar GCC}

A quantizable
	combinatorial game~$\G$
	is bisimilar to its classical versus classical quantum variant,
	$\G \sim \G^{C\text{--}C} := \G^{Q|\emptyset}$.
\end{proposition}

\subsection{Supremacy of a perfect classical player}\label{subsect:supremacy}


Besides a definition of a quantum variant as a generalization of the classical game, 
we are interested in studying the advantage of a quantum player over a classical player.
Surprisingly,
it turns out that if one is able to play perfectly,
there is no additional benefit to having access to quantum moves,
see Theorem~\ref{no-quantum-advantage} below.

\begin{theorem}\label{no-quantum-advantage}
Given a quantizable combinatorial game~$\mathcal{G}$,
	and a position~$\ket{\varphi}$ in any
	partially quantum variant~$\G^{Q|P}$
	(think $\G^{Q\text{--}C}$, $\G^{C\text{--}Q}$, $\G^{C\text{--}C}$ or $\G^Q$ itself).
	Then for any `target outcome' $\mathcal{T}\subseteq\Omega$,
\vspace{-2.5em}

\begin{alignat}{3}
\label{nqa-1}
\Supp{}{\ket \varphi}\,\subseteq\,\Sigma_{\mathcal{T}}^C
	\qquad&\implies\qquad \ket{\varphi}\,\in\, 
	\Sigma_{\mathcal{D}(\mathcal{T})}^{Q|P} \\
\label{nqa-2}
\Supp{}{\ket \varphi}\,\subseteq\,\Pi_{\mathcal{T}}^C
	\qquad&\implies\qquad \ket{\varphi}\,\in\, 
	\Pi_{\mathcal{D}(\mathcal{T})}^{Q|P}
\end{alignat}
\end{theorem}
\begin{proof}
We prove by induction on the moves left in position~$\ket\varphi$
that statements~\eqref{nqa-1}
	and~\eqref{nqa-2} hold for all~$\mathcal{T}$.

Suppose that~$\ket\varphi$ has~$0$ moves left.
Then~$\ket\varphi$ is definitely finished,
hence all~$g\in \Supp{}{\ket\varphi}$ are finished.
Let~$\mathcal{T}$ be given,
and suppose that~$\Supp{}{\ket\varphi} \subseteq \Sigma_{\mathcal{T}}^C$.
Then~$\mathcal{O}(g)\in\mathcal{T}$
for all~$g\in \Supp{}{\ket\varphi}$,
and thus the distribution~$\mathcal{O}(\ket\varphi)$
is only non-zero on~$\mathcal{T}$,
so $\mathcal{O}(\ket\varphi)\in\mathcal{D}(\mathcal{T})$.
Thus 
$\ket\varphi \in \smash{\Sigma_{\mathcal{D}(\mathcal{T})}^{Q|P}}$.
Hence~\eqref{nqa-1} holds for~$\ket\varphi$.
The proof that~\eqref{nqa-2} holds for~$\ket\varphi$ is almost identical.

Now let~$\ket\varphi$ be again arbitrary,
and suppose that $\ket\varphi$ 
has~$N+1$ moves left.
Further suppose
that~\eqref{nqa-1} and~\eqref{nqa-2}
	hold for any~$\mathcal{T}$, and  any~$\ket{\varphi'}$
	with~$N$ moves left.
We'll show that~\eqref{nqa-1} and~\eqref{nqa-2}
hold for~$\ket\varphi$ as well.
First note that if~$\ket\varphi$ is finished,
then~\eqref{nqa-1} and~\eqref{nqa-2} hold for the reasons mentioned above,
so we may assume without loss of generality that~$\ket\varphi$ is not finished.

\emph{Statement~\eqref{nqa-1}.}
Let~$\Supp{}{\ket\varphi}\subseteq \Sigma^C_{\mathcal{T}}$.
This means that for each~$g\in \Supp{}{\ket\varphi}$
there is~$\hat{g}\in \Pi^C_{\mathcal{T}}$
such that $g\xrightarrow{\sigma} \hat{g}$ for some~$\sigma$.
Note that this is true  also when~$g$ is already finished,
because the classical game is fixed-length,
so there is still a classical move from~$g$,
and such a move on a finished position does not change the outcome,
which thus remains in~$\mathcal{T}$.



Note that there is a unitary $U$ with $U\ket g = \ket{\hat g}$ for all $g \in \Supp{}\ket\varphi$.
Indeed, for each~$g\in\Supp{}\ket\varphi$ pick~$\sigma_g$ with~$g\smash{\xrightarrow{\sigma_g}}\hat g$,
and apply Lemma~\ref{lem: all classical moves have an equivalent quantum move}
to the map~$\omega\colon g\mapsto \sigma_g$ on~$S:=\Supp{}\ket\varphi$. 
Note that this unitary~$U$ is a valid move on~$\ket\varphi$
because for each~$g\in \Supp{}\ket\varphi$
there is a $\sigma$
with $g\smash{\xrightarrow{\sigma}}\hat{g}$,
by definition of~$\hat{g}$.
Moreover, $U$ is a classical move by Lemma \ref{lem: all classical moves have an equivalent quantum move},
which makes~$U$ a valid move in $\G^{Q|P}$,
whether or not~$\ket\varphi \in P$.
We claim that~$U\ket\varphi \in 
\smash{\Pi^{Q|P}_{\mathcal{D}(\mathcal{T})}}$
which implies that $\ket\varphi \in 
\smash{\Sigma^{Q|P}_{\mathcal{D}(\mathcal{T})}}$.
By the induction hypothesis, 
and specifically~\eqref{nqa-2},
it suffices to show that $\Supp{}(U\ket\varphi) \subseteq \Pi_{\mathcal{T}}^C$.
Moreover, by Lemma~\ref{supp-lemma}
it suffices to show that $\Supp{}(U\ket g) \subseteq \Pi_{\mathcal{T}}^C$
given~$g\in \Supp{}\ket{\varphi}$.
But  $U\ket{g}=\ket{\hat{g}}$,
so we just have to show that $\hat{g} \in \Pi_{\mathcal{T}}^C$,
but this we already know.
Whence~\eqref{nqa-1} holds for~$\ket\varphi$.

\emph{Statement~\eqref{nqa-2}.}
Suppose that $\Supp{}{\ket\varphi}\subseteq \Pi_{\mathcal{T}}^C$.
We must show 
that~$\ket\varphi \in \smash{
\Pi_{\mathcal{D}(\mathcal{T})}^{Q|P}}
$.
That is, 
given a valid move~$U$ on~$\ket\varphi$
we must show that 
$U\ket{\varphi}\in \smash{
	\Sigma^{Q|P}_{\mathcal{D}(\mathcal{T})}}$.
By the induction hypothesis,
and specifically~\eqref{nqa-1},
it suffices to show that $\Supp{}(U\ket\varphi)\subseteq \Sigma^C_{\mathcal{T}}$.
So let~$g'\in \Supp{}(U\ket\varphi)$ be given.
Since~$U$ is a valid move on~$\ket\varphi$,
there is some~$g\in \Supp{}\ket\varphi$ and~$\sigma$
with~$g\smash{\xrightarrow{\sigma}} g'$.
Since $g\in \Supp{}\ket\varphi$
and we assumed
$\Supp{}\ket\varphi \subseteq \Pi_{\mathcal{T}}^C$,
this implies that~$g'\in\Sigma_{\mathcal{T}}^C$.
Hence $\Supp{}(U\ket\varphi)\subseteq \Sigma^C_{\mathcal{T}}$,
and so~\eqref{nqa-2} holds for~$\ket\varphi$.
\end{proof}

Since the quantum variant of a combinatorial game is also a combinatorial game, Zermelo's theorem
(Corollary~\ref{cor: zermelo's result})
also holds for this quantum variant. However, we remark that this is not a very strong theorem in the quantum case. Since the outcome space of the quantum variant is the set of distributions on the classical outcome set, the partitions as in Zermelo's theory contain uncertainty about the outcomes. A priori it is not clear whether we can partition the outcome set such that the guaranteed outcome does not contain distributions with uncertainty on the outcome. In the following corollary, we prove a stronger version of Zermelo's theorem which states that we can ensure a deterministic outcome and moreover that this outcome coincides with the classical variant of the game. So, for example, if a win can be guaranteed for player X in the classical game, then it can also be guaranteed in the quantum variant. 

We have Corollary~\ref{cor: zermelo for quantum} whose proof
follows immediately from Corollary~\ref{cor: zermelo's result} and 
	Theorem~\ref{no-quantum-advantage}.

\begin{corollary}\label{cor: zermelo for quantum}
	Given a position~$g$ with $\turn(g)=X$ in a quantizable combinatorial
game~$\G$, and a partially quantum variant $\G^{Q|P}$ of~$\G$,
and a partition $\Omega = \Omega_X\sqcup \Omega_O\sqcup \Omega_T$ of the outcome set of the classical game.
Then exactly one of the following three mutually exclusive cases holds.
	\begin{enumerate}
		\item $g\in\Sigma_{\Omega_X}$, and in that case $\ket g \in \Sigma^{Q|P}_{\mathcal{D}(\Omega_X)}$.
		\item $g\in \Pi_{\Omega_O}$, and in that case
			$\ket g \in \Pi^{Q|P}_{\mathcal{D}(\Omega_O)}$.
		\item $g\in \Sigma_{\Omega_X\sqcup \Omega_T}$
			and $g\in \Pi_{\Omega_O\sqcup \Omega_T}$,
			and in that case $\ket g \in \Sigma^{Q|P}_{\mathcal{D}(\Omega_X \cup \Omega_T)}$ and $\ket g \in \Pi^{Q|P}_{\mathcal{D}(\Omega_O \cup \Omega_T)}$.
	\end{enumerate}
\end{corollary}
\begin{remark}
We remark that just like with Zermelo's theorem this statement holds in the particular case where $\ket g= \ket 0$ and that in this case --- since $0 \in \G$ --- we get that a classical game has some guaranteed outcome if and only if its quantum variant has the same outcome. This means that the outcomes
a player can achieve (using perfect play) do not change when quantum moves
are added to the mix, in any way.
That is, 
a classical player (making all the right moves) is
just as strong as their quantum counterpart.
\end{remark}

\subsection{Playing against an imperfect opponent}
The take-away from Corollary~\ref{cor: zermelo for quantum}
is that quantum moves
give no benefit against a classical player playing perfectly.
Perfect play can be 
expected for a simple game like tic-tac-toe,
but who knows which opening moves of chess are winning,
and which ones are not?
Interestingly,
if a classical player is not able to play perfectly,
a quantum player can fish for a mistake from its classical opponent by 
making superpositions of moves that it considers strong.
When at some point it becomes clear that one of the components
of the superposition is winning for the quantum player,
they can `amplify' this particular classical position (see Lemma~\ref{amplify-mistake})
using what is technically a valid move in the quantum variant.
This means that when playing against a quantum opponent,
the classical player can't afford to make any mistakes.
\begin{lemma}[Flaw amplification]\label{amplify-mistake}
Given a position $\ket\varphi$ in the quantum variant~$\G^Q$
of a combinatorial game~$\G^C$
and a target outcome $\mathcal{T}\subseteq\Omega$.
If $g\in \Sigma_{\mathcal{T}}^C$ and~$g$ is not finished
for some~$g\in \Supp{}{\ket\varphi}$,
then $\ket\varphi \in \Sigma_{\mathcal{D}(\mathcal{T})}^Q$.
\end{lemma}
\begin{proof}
Since $g\in \Sigma_{\mathcal{T}}^C$
and~$g$ is not finished,
there is some move~$g \smash{\xrightarrow{\sigma}} h$
with $h\in \Pi_{\mathcal{T}}^C$.
By Theorem~\ref{no-quantum-advantage}
	we have $\ket{h}\in\Pi_{\mathcal{D}(\mathcal{T})}^Q$.
So let~$U$ be a unitary with $U\ket\varphi = \ket{h}$.
To create such a unitary,
extend $\ket\varphi$ and $\ket{h}$ 
to orthonormal bases $\ket{e_1}:=\ket\varphi,\,\ket{e_2},\,\dotsc,\,\ket{e_{N}}$
and $\ket{f_1}:=\ket{h},\,\ket{f_2},\,\dotsc,\,\ket{f_{N}}$, respectively,
	where $N=2^L$ (see Definition~\ref{qv} for ``$2^L$''),
and let~$U$ be defined by  $U\ket{e_n} = \ket{f_n}$ for all~$n$.
This unitary~$U$ is a valid move on~$\ket\varphi$,
	for the simple reason that $\Supp{}{(U\ket\varphi)}\equiv\{h\}$
and $g\smash{\xrightarrow{\sigma}} h$.
	Whence $\ket\varphi \in \Sigma_{\mathcal{D}(\mathcal{T})}^Q$.
\end{proof}
\begin{remark}
\label{completely-valid}
One could argue that this reveals
not a cunning strategy for the quantum player,
but perhaps that too many unitaries are considered valid quantum moves
according to our Definition~\ref{qv}.
The following notion would arguably have been more natural given the linearity
of quantum mechanics.
Let us call a unitary~$U$ \textbf{completely valid}
on~$\ket\varphi$ if $\ket\varphi$ is not finished
and for each~$g\in \Supp{}{\ket\varphi}$,
$U\ket{g}$ is a superposition of positions in~$\{g'\in\G^C\colon g\xrightarrow{\sigma} g' \text{ for some } \sigma\in\Sigma^C\}$.
Any such completely valid quantum move
is also valid in the regular sense (by Lemma~\ref{supp-lemma}).
Moreover, the classical moves constructed
via Lemma~\ref{lem: all classical moves have an equivalent quantum move} are completely valid.
This means that our main Theorem~\ref{no-quantum-advantage}
also holds if we opt for \emph{completely} valid moves
in the definition of the quantum variant~$\G^Q$.
However, this restricts the power of the quantum player,
and thus makes Theorem~\ref{no-quantum-advantage} weaker.
Hence we did not opt for this more restricted notion of quantum move.
However,
if one ever gets the chance to play a game of chess on a quantum computer
as a classical player,
one should definitely demand that all moves be completely valid,
precisely because the trick above cannot be played against you.
\end{remark}

\section{Applications: Quantum Byzantine Consensus}\label{sec:applications}
 A fundamental problem in distributed networks is that of guaranteeing reliability of the whole system even in the case of failure or adversarial behaviour of some of its nodes. Often, this reliability needs to be preserved even if the adversarial nodes communicate false or contradictory information to the rest of the nodes. This is often modelled as achieving Byzantine consensus -- a reference to the Byzantine generals problem~\cite{lamport2019byzantine}. 
 Here, consensus means that: 
\begin{enumerate}
	\item All non-Byzantine (= well-behaving) nodes agree on the same decision value (\textbf{Agreement}).
	\item If all non-Byzantine nodes had the same initial decision value, then the final decision value must be the same as the initial decision value (\textbf{Validity}).
\end{enumerate}
Interestingly, this problem can be modelled as a classical combinatorial game. Here we give a sketch of the idea, and the more detailed modelling can be found in Appendix~\ref{ap:QBC}. 

Our model uses the standard approach of modelling distributed networks as graphs, with each node having ``internal'' state, ``incoming'' and ``outgoing'' messages and ``decision'' value. Rounds are synchronous, and in each, all nodes do a state transition, with input the ``internal'' state and  ``incoming'' message which changes the ``internal'' state and the ``outgoing'' messages (which become incoming in the next round).
In our model, player X represents the well-behaving nodes and player O represents the Byzantine nodes. Player X has a strategy to play, which can be chosen at the beginning of the game, or hardcoded in the starting position. During each round, player X first delivers all messages and applies its chosen algorithm to all well-behaving nodes, after which player O applies some move to each of the Byzantine nodes. The winner of the game is determined by looking at all decision values of the well-behaving nodes after the last round and seeing whether they reach consensus. Note that there is no tying in this game.
For the quantum variant, the only meaningful version is where player X plays classically and player O plays quantumly. It tells us that given a topology and a number of Byzantine failure nodes if there exists a classical algorithm to guarantee consensus against classical failure nodes, then there also exists a classical algorithm to ensure consensus against quantum failure nodes. Furthermore, if we specify the classical algorithm used in the initial position we get that the algorithm that ensures consensus against the classical Byzantine nodes is precisely the algorithm that guarantees consensus against the quantum Byzantine nodes.

Our model of a combinatorial game of the Byzantine consensus might not be the most efficient. Indeed, Lamport's solution~\cite{lamport2019byzantine} that can be seen as a combinatorial game requires a significant number of message passes. Nevertheless, if one plays it as a combinatorial game, our results are meaningful, and guarantee that there is no quantum advantage against honest parties that have a strategy that can classically always lead to consensus. This reduced efficiency is an acceptable trade-off for the strong guarantee it offers.

Note that the quantum variant of the Byzantine agreement problem has already been treated in the literature:
Ben-Or et al.~\cite{ben2005fast} provide a more efficient solution for honest parties, while Fitzi et al.~\cite{fitzi2001quantum} show an advantage for quantum well-behaving nodes in certain settings.
In contrast, we provide a model that guarantees that a classical solution can be transferred to the quantum setting.

\section{Conclusion and outlook}\label{sec:conclusion}
We have defined a quantum variant for a broad class of combinatorial games,
including tic-tac-toe, chess, and checkers,
and shown 
that perfect classical play remains invincible against a quantum opponent
(Corollary~\ref{cor: zermelo for quantum}).
As an application, we saw classical Byzantine consensus algorithms
that tolerate classical adversaries equally tolerate quantum ones.
In practice, flawless play is not realistic,
and the quantum player \emph{can} exploit any mistake
via flaw amplification (Lemma~\ref{amplify-mistake}).
This, however, relies on our permissive definition of valid quantum move;
under the stricter notion of completely valid moves,
this specific amplification technique does not apply.

Although formulated for two-player combinatorial games without chance,
we expect the proof of Corollary~\ref{cor: zermelo for quantum} to generalize to any number of players.
Moreover, by modelling chance and probabilistic strategies
as moves by an additional player
(e.g.~`chance', see~\cite[Section~6.3.1]{OsborneRubinstein1994}),
we believe the result extends to these settings as well:
any outcome that a classical player can guarantee
would be preserved against a quantum opponent.
Conversely, Meyer's PQ Penny Flip game~\cite{Meyer1999} demonstrates
that imperfect information can enable a quantum player to guarantee outcomes
that a classical player cannot.
In summary, quantum advantage is real, but elusive.

\section{Acknowledgements}
We'd like to thank the anonymous reviewers for their valuable feedback.
One reviewer in particular gave very detailed and helpful suggestions.

\bibliographystyle{alphaurl} 

\bibliography{references}

\appendix
\section{Omitted theory and proofs (Classical combinatorial games)}
\label{omitted-proofs}
\label{omitted-proofs-classical}
The maps $\mathcal{T}\mapsto \Sigma_{\mathcal{T}}$ and $\mathcal{T}\mapsto \Pi_{\mathcal{T}}$
are order preserving:
\begin{lemma}\label{lem:A subset B then sigma_A subset sigma_B and for Pi}
Given some combinatorial game $\mathcal{G}$, and  $\mathcal{S} \subseteq \mathcal{T} \subseteq \Omega$,
we have:
	$$
	\Sigma_{\mathcal{S}} \,\subseteq\, \Sigma_{\mathcal{T}}\qquad\text{and}\qquad \Pi_{\mathcal{S}}\,\subseteq\,\Pi_{\mathcal{T}}
	$$
\end{lemma}
\begin{proof}
Let~$\mathcal{S},\mathcal{T}\subseteq \Omega$ be given.
	We prove by induction on~$N(g)$
	that for any~$g\in 
	\mathcal{G}$,
\begin{equation}
\label{sigpiop}
g\in \Sigma_\mathcal{S} \implies g\in \Sigma_\mathcal{T}
	\qquad\text{and}\qquad
g\in \Pi_\mathcal{S} \implies g\in  \Pi_\mathcal{T}.
\end{equation}
Suppose that~$N(g)=0$.
Then~$g$ is finished.
So if $g\in \Sigma_{\mathcal{S}}$ we must have~$g\in \mathcal{S}\subseteq \mathcal{T}$,
and thus $g\in \Sigma_{\mathcal{T}}$.
	Similarly,  $g\in \Pi_{\mathcal{S}}$ implies
	$g\in \Pi_{\mathcal{T}}$.

Next suppose that~$g\in \mathcal{G}$
with $N(g)>0$
is given such that~\eqref{sigpiop} holds
for all~$g'$ with $N(g')<N(g)$.
Suppose further
	that~$g\in \Sigma_{\mathcal{S}}$.
	Then $\smash{g\xrightarrow{\sigma} g'}$
for some~$\sigma\in\Sigma$ and~$g'\in \Pi_{\mathcal{S}}$.
	Then~$g'\in \Pi_{\mathcal{T}}$ by the induction hypothesis,
and so $g\in\Sigma_{\mathcal{T}}$.
With a similar reasoning we get 
	that $g\in \Pi_{\mathcal{S}}$ implies $g\in \Pi_{\mathcal{T}}$.
Hence~\eqref{sigpiop} holds for~$g$.
\end{proof}
\begin{proof}[of \emph{\textbf{Corollary~\ref{cor: zermelo's result}}}]
Note that either $g \in \Sigma_{\Omega_X}$ (case 1) holds or not,
$g \not\in \Sigma_{\Omega_X}$,
which (by Theorem~\ref{thm:fund-cgt}) is equivalent to
$g \in \Pi_{\Omega \backslash \Omega_X} = \Pi_{\Omega_T \sqcup \Omega_O}$.
Similarly, either $g\in \Pi_{\Omega_O}$ (case 2) holds, or otherwise
$g \in \Sigma_{\Omega_T \sqcup \Omega_X}$ (but not both).
Note that if neither case~1 nor case~2 holds, we are in case~3.
A priori
there is the option that both  case~1 and case~2 hold,
but this does not actually happen.
If $g\in \Sigma_{\Omega_X}$ (case 1),
then also $g\in \Sigma_{\Omega_T \sqcup \Omega_X}$
(by~Lemma~\ref{lem:A subset B then sigma_A subset sigma_B and for Pi}),
which means that case~2 does not hold by the argumentation above.
\end{proof}
\begin{proof}[of \emph{\textbf{Proposition~\ref{prop:guaranteeable sets are equivalent for bisimilar games}}}]
We only have to prove the equivalence involving~$\Sigma_{\mathcal{T}}$;
	the equivalence involving~$\Pi_{\mathcal{T}}$
	follows from it by Theorem~\ref{thm:fund-cgt}.
We use induction over~$N(g_1)$. Let $g_1$ be given.

	Suppose that~$N(g_1)=0$,
	and let~$g_2$ with $g_1 \sim g_2$ be given.
Then~$g_1 \in \mathcal{F}$,
and thus~$g_2 \in \mathcal{F}$ and $\mathcal{O}(g_1)=\mathcal{O}(g_2)$
(since $g_1\sim g_2$).
Now,
we have $g_1 \in \Sigma_{\mathcal{T}}$ 
iff $\mathcal{O}(g_1)\equiv \mathcal{O}(g_2) \in \mathcal{T}$
iff $g_2 \in \Sigma_{\mathcal{T}}$.

Suppose now that~$N(g_1) > 0$,
and that~\eqref{bisimilarity-implies-equistrategy}
already holds
for all~$\tilde{g}_1$ with $N(\tilde{g}_1) < N(g_1)$
 (and any~$g_2$).
We'll show that~\eqref{bisimilarity-implies-equistrategy}
holds for~$g_1$ as well.

Let $g_2$ with $g_1\sim g_2$ be given.
Note that since~$N(g_1)>0$,
we know that~$g_1$ is not finished,
and since~$g_1 \sim g_2$,
we know~$g_2$ is not finished either.

Suppose that $g_1 \in \Sigma_{\mathcal{T}}$.
We must show that $g_2\in\Sigma_{\mathcal{T}}$ as well.
Since~$g_1$ is not finished and  $g_1\in \Sigma_{\mathcal{T}}$
	there is a move
	$g_1 \xrightarrow{\sigma_1} g_1'$
	such that~$g_1' \in \Pi_{\mathcal{T}}$.
Since~$g_1\sim g_2$
	there is a move $g_2 \xrightarrow{\sigma_2} g_2'$
	such that $g_1'\sim g_2'$.
Since~$N(g_1')<N(g_1)$,
	we know that equation~\eqref{bisimilarity-implies-equistrategy}
	already holds for $N(g_1')$.
	In particular, $g_2' \in \Pi_{\mathcal{T}}$.
	This means that from~$g_2$ we can move into~$\Pi_{\mathcal{T}}$,
	so $g_2\in\Sigma_{\mathcal{T}}$.

Suppose instead that~$g_2 \in \Sigma_{\mathcal{T}}$;
	we must show that $g_1 \in \Sigma_{\mathcal{T}}$.
	Since~$g_2$ is not finished, we know there is some
move $g_2 \xrightarrow{\sigma_2} g_2'\in \Pi_{\mathcal{T}}$.
Since~$g_1\sim g_2$,
and~$g_2$ is not finished,
	there is some move $g_1 \xrightarrow{\sigma_1} g_1'$
	with $g_1'\sim g_2'$.
Thus~$g_1'\in \Pi_{\mathcal{T}}$ 
since~\eqref{bisimilarity-implies-equistrategy} already holds for~$g_1'$.
Thus $g_1\in \Sigma_{\mathcal{T}}$.
\end{proof}

\section{Making a game injective and of fixed length}
\label{injective-and-fixed-length}
To make a game injective,
we employ the following construction.
\begin{definition}\label{def:transcribed games}
The \textbf{transcribed variant} $\Grec$ of a combinatorial game~$\G$ is defined as follows.
\begin{itemize}
\item
As \emph{positions}, $\Grec$ has
\textbf{transcripts},
pairs $(g,R)\in \G\times \Sigma^*$,
where $g\in\G$ and $R\equiv (R_1,\dotsc,R_n)$
is a sequence of moves in~$\G$ from~$0$ to~$g$,
so
		$0 \smash{\xrightarrow{R_1} g_1 \xrightarrow{R_2} \dots \xrightarrow{R_n}} g$.
\item
	$\Grec$ has the same set of \emph{moves},  $\Sigma$,
		as $\G$,
		and its \emph{transition map}
		is given by
	$$
		\Gammarec (\,(g, (R_1,\dotsc,R_n)), \,\sigma\,)\  =\  \begin{cases}
			(\,\Gamma(g, \sigma), \,(R_1,\dotsc,R_n,\sigma)\,) &\text{if } \Gamma(g, \sigma )\neq \perp,\\
		\perp &\text{otherwise.}
	\end{cases}
	$$
The \emph{starting position} of $\Grec$ is simply $(0,())$.

\item
$\Grec$ 
has the same \emph{outcome set}, $\Omega$, as~$\G$.
A position $(g,R)$ in $\Grec$ is finished when~$g$ is finished in~$\G$,
and in that case 
the outcome is given by~$\mathcal{O}((g,R))=\mathcal{O}(g)$.
\item
Finally, $\mathrm{turn}((g,R))=\mathrm{turn}(g)$.
\end{itemize}
\end{definition}
The proof that $\Grec$ obeys the six conditions
we put on a combinatorial game (see Definition~\ref{cg})
is easy,
and omitted.
What's more important is that~$\Grec$ is injective:
\begin{lemma}\label{lem: a transcribed game is injective}%
The transcribed variant~$\Grec$ of a combinatorial game~$\G$
is injective.
\end{lemma}
\begin{proof}
	Let $(g, R), (h, S) \in \Grec$ and $\sigma, \pi \in \Sigma$ such that $\Gammarec((g, R), \sigma) = \Gammarec((h, S), \pi) \neq \perp$
	 be given.  We must show that $(g,R)=(h,S)$ and $\sigma=\pi$.

Write $R\equiv(R_1,\dotsc,R_n)$ and $S\equiv (S_1,\dotsc,S_m)$.
Then by definition of  $\Gammarec$ we get 
$$(\Gamma(g, \sigma), (R_1,\dotsc,R_n,\sigma))\ =\ (\Gamma(h, \pi), (S_1,\dotsc,S_m,\pi)),$$
	which implies that
	$\sigma=\pi$ and $R = S$.
Since therefore~$g$ and~$h$ are reached starting at~$0$ by the same sequence of moves~$R=S$,
	we have $g=h$.
	Hence~$\Grec$ is injective.
\end{proof}

A combinatorial game can be made  fixed length as follows.

\begin{definition}\label{def: extended game}
	The \textbf{extended variant} $\G_{ext}$
	of a comb.~game~$\G$ is
	defined as follows.
\begin{itemize}
\item
The game~$\G_{ext}$
basically has the same \emph{positions} as~$\G$, except that finished
positions are duplicated so that either~$X$ or~$O$ can be the player-to-move
on them:
 $$\G_{ext} \ = \ 
\{ \,(g,\turn(g))\,\colon\,g\in \G\backslash\mathcal{F}\,\}\ 
\cup\ \{\,(g,P)\,\colon\, g\in\mathcal{F},\, P\in \{X,O\}\,\}
\ \subseteq\ \G\times \{X,O\}.$$
		The \emph{starting position} is~$(0,X)$,
		and on a position $(g,P)$, the player-to-move is just $\turn((g,P))=P$.
\item
$\G_{ext}$ adds one additional \emph{move}, $\bullet$, (the `extension move') to~$\G$,
so $\Sigma_{ext}=\Sigma\cup\{\bullet\}$.
This extension move can be made on any finished position~$g\in\mathcal{F}$,
and will only switch the player, so
$$\Gamma((g,P),\bullet)\ =\  (g,\overline{P}),$$
where $\overline{X}=O$ and~$\overline{O}=X$.
For a move $\sigma \in \Sigma$ from the original game~$\G$,
the \emph{transition map}  is given by 
$\Gamma((g,P),\sigma)\ =\  (\,\Gamma(g,\sigma),\,\overline{P}\,).$
\item
$\G_{ext}$ has the same outcome set~$\Omega$ as~$\G$.
A position $(g,P)$ is \emph{finished}
when~$g$ is finished,
and in that case $\mathcal{O}((g,P)) = \mathcal{O}(g)$.
\end{itemize}
\end{definition}

That the extended variant of a combinatorial game itself
obeys the six conditions
of Definition~\ref{cg} is straightforward to check.
Recall that the ``games end'' condition~\ref{games-end} 
only requires that any sequence of~$N$ moves results in a finished position.
Nothing in our definition of combinatorial game prohibits moves on finished positions, 
like the extension move $\bullet$ (as long as the outcome stays the same.)

What's special is that~$\G_{ext}$ is also  fixed length:
\begin{lemma}\label{lem:an extended is fixed-length}
The extended variant $\G_{ext}$ of a combinatorial game~$\G$ is  fixed length.
\end{lemma}
\begin{proof}
There is a number~$N$ for the game~$\G$ 
such that any sequence of~$N$ moves in~$\G$ from~$0$ results
in a finished position.
Clearly,
the same will be true for this~$N$ in~$\G_{ext}$.
But in~$\G_{ext}$
we also have that
in any position $(g,P)$
there is still a move that can be made.
	Indeed, either $(g,P)$ (and thus~$g$) is not finished,
	and therefore there must already have been a  move from~$g$ in~$\G$
	(which can be made on $(g,P)$ in~$\G_{ext}$ too),
	\emph{or} $(g,P)$ (and thus~$g$) is finished,
	and we can thus make the extension move $(g,P)\xrightarrow\bullet(g,\overline{P})$.
\end{proof}

To make combinatorial game
both injective and of fixed-length,
one first adds the extension move,
and then adds the transcripts, see the proposition below.
The other way around does not work:
the extended variant of a combinatorial game is generally not injective.
\begin{proposition}\label{prop: the transcribed extended game is injective and fixed-length}
The transcribed variant~$\Grec$ of a fixed-length combinatorial game~$\G$ is of fixed-length too.
In particular,
$(\G_{ext})_{tr}$
is both injective and  fixed length.
\end{proposition}
\begin{proof}
Since~$\G$ is  fixed length, there is some number~$N$
such that any sequence of~$N$ moves from~$0$
results in a finished position in~$\G$,
and after any~$M<N$ moves from~$0$
a move can still be made.
It's clear that a sequence of~$N$ moves in~$\Grec$ results in a finished position too.
Suppose $(g,R)$ is reachable from~$(0,())$
in~$M<N$ moves.
Then the exact same moves can be made in~$\G$ too
to reach $g$ from~$0$.
In~$g$, there is a still a move that can be made
and that move can be made on~$(g,R)$ too,
by definition of the transition map on~$\Grec$.
Hence~$\Grec$ is  fixed length.
\end{proof}

We now know how to make a
combinatorial game
both
injective and  fixed length,
which are required for our definition of
quantum variant (see Definition~\ref{qv}).
When we speak, for example,
of quantum tic-tac-toe,
we do not  actually mean the quantum variant of tic-tac-toe itself (which is not defined),
but of the transcribed variant of the extended variant of tic-tac-toe.
This extended and then  transcribed variant
is
in gameplay, however, indistinguishable
from the original variant,
which we formalize below using the  notion of bisimulation\cite{bisimulation-history}.

With this notion of bisimilarity between games we can show that transcribing 
or extending a game has no effect on gameplay.
\begin{lemma}\label{lem: a game is bisimilar to its transcribing}
A combinatorial game~$\G$ 
and its transcribed variant~$\Grec$
are bisimilar.
\end{lemma}
\begin{proof}
Define $\sim \subseteq \G \times \Grec$ by $g \sim (h, R)$ iff $g=h$. Then first of all,
clearly
$0 \sim (0, ()) = 0_{tr}$. Now
	let $g\sim (g,(\sigma_1,\dotsc,\sigma_n))$.
Then~$\sim$ is a bisimulation,
because, using the same numbering of conditions as in  Definition~\ref{def:bisimulation}:
	\begin{enumerate}
		\item We have $g \in \mathcal{F}$ iff $(g, R) \in\mathcal{F}_{tr}$,
			and in that case  $\outcome(g) = \outcome((g,R))$, by definition.
		\item Suppose $g \not\in \mathcal{F}$.
		\begin{enumerate}
			\item If $g \xrightarrow{\sigma} g'$,
				then we have the corresponding move
				$(g, (\sigma_1,\dotsc,\sigma_n)) \xrightarrow{\sigma} (g', (\sigma_1,\dotsc,\sigma_n,\sigma))$
				with $g'\sim (g',(\sigma_1,\dotsc,\sigma_n,\sigma))$.
			\item 
		If $(g, (\sigma_1,\dotsc,\sigma_n)) \xrightarrow{\sigma} (g', (\sigma_1,\dotsc,\sigma_n,\sigma))$ then we have the corresponding move $g \xrightarrow{\sigma}  g' \sim (g', (\sigma_1,\dotsc,\sigma_n,\sigma))$.
		\end{enumerate}
	\end{enumerate} 
\end{proof}

\begin{lemma}\label{lem: a game is bisimilar to its extension}
A combinatorial game~$\G$ and its extended variant $\G_{ext}$ are bisimilar.
\end{lemma}
\begin{proof}
Define~$\sim\subseteq \G\times\G_{ext}$ by $g\sim(h,P)$
iff $g=h$.
Then first of all, clearly~$0\sim(0,X)\equiv 0_{ext}$.
	Now, let $g\sim(g,P)$ be given.
The relation~$\sim$
is a bisimulation,
because, using the same numbering of conditions as in Definition~\ref{def:bisimulation}:
\begin{enumerate}
\item We have $g \in \mathcal{F}$ iff $(g, P) \in \mathcal{F}_{ext}$,
	and in that case, $\mathcal{O}((g,P))\equiv \mathcal{O}(g)$, by definition.

\item Suppose $g \not\in \mathcal{F}$
\begin{enumerate}
\item When $g \xrightarrow{\sigma} g'$ for some 
$\sigma \in \Sigma$, 
we have a corresponding move $(g,P)\xrightarrow\sigma (g',\overline{P})$,
with $ g'\sim (g',\overline{P})$.
\item When $(g,P)\xrightarrow\sigma (g',P')$,
	we know~$\sigma\neq \bullet$ (since $g\notin \mathcal{F}$),
		so we have a corresponding move
		$g\xrightarrow\sigma g'$ in~$\G$,
		with $g'\sim(g',P')$.
		\end{enumerate}
	\end{enumerate}
\end{proof}

Whence $\G \sim \G_{ext}\sim (\G_{ext})_{tr}$.
To conclude that~$\G\sim (\G_{ext})_{tr}$,
we need the fact that bisimilarity is a transitive relation on games.
\begin{lemma}\label{lem: transitivity of bisimilarity relation}
	The bisimulation relation $\sim$ is transitive, i.e.\
	if $\G_1 \sim \G_2$ and $\G_2 \sim \G_3$ then $\G_1 \sim \G_3$, for combinatorial games $\G_1, \G_2$ and $\G_3$.
\end{lemma}
\begin{proof}
Let the bisimulations between $\G_1$ and~$\G_2$ and between $\G_2$ and~$\G_3$
both be denoted by~$\sim$.
We define the bisimilarity relation between $\G_1$ and $\G_3$ by
 $g_1 \sim g_3$ iff $g_1 \sim g_2$ and $g_2 \sim g_3$.

	We must now show that this relation is indeed a bisimulation. First, $0_1 \sim 0_2$ and $0_2 \sim 0_3$ so $0_1 \sim 0_3$. Then, we go over the conditions as in Definition~\ref{def:bisimulation}.

	Let $g_1 \in \G_1$ and $g_3 \in \G_3$ such that $g_1 \sim g_3$, then by definition there exists a $g_2 \in \G_2$ such that $g_1 \sim g_2 \sim g_3$. 
	\begin{enumerate}
		\item $g_1 \in \mathcal{F}_1$ iff $g_2 \in \mathcal{F}_2$ iff $g_3 \in \mathcal{F}_3$, and then $\outcome(g_1) = \outcome(g_2) = \outcome(g_3)$.
		\item Assume $g_1 \not \in \mathcal{F}_1$, then:
		\begin{enumerate}
			\item If $g_1 \xrightarrow{\sigma_1} g_1'$, then $g_2 \xrightarrow{\sigma_2} g_2'$ for some $\sigma_2$ and $g_2'$, with $g_1' \sim g_2'$. Since $\G_2 \sim \G_3$ we also have that $g_3 \xrightarrow{\sigma_3} g_3'$ for some $\sigma_3$ and $g_3'$ with $g_2' \sim g_3'$.
				This gives that $g_1' \sim g_3'$.
			\item This case is analogous to (a).
		\end{enumerate}
	\end{enumerate}
	So, $\G_1 \sim \G_3$.
\end{proof}

\begin{corollary}
	\label{modified-game-bisimilar}
	Any combinatorial game $\G$ is bisimilar
	to its extended and then transcribed variant, $\G \sim (\G_{ext})_{tr}$.
\end{corollary}
\begin{proof}
	This follows directly from Lemmas \ref{lem: a game is bisimilar to its transcribing}, \ref{lem: a game is bisimilar to its extension}, and \ref{lem: transitivity of bisimilarity relation}.
\end{proof}

\section{Omitted theory and proofs (Quantum combinatorial games)}
\label{omitted-proofs-quantum}
\begin{lemma}
\label{supp-lemma}
Given a move~$U$ on a position~$\ket{\varphi}$ in the quantum variant
of a quantizable combinatorial game~$\mathcal{G}$,
	we have
$\Supp{}(U\ket\varphi) \ \subseteq \ \bigcup_{g \in \Supp{}\ket \varphi}\Supp{}(U\ket g)$.
\end{lemma}
\begin{proof}
Let~$s\in \Supp{}(U\ket\varphi)$ be given.
This means that $\left<s \right| U \ket\varphi \neq 0$.
	Since $\ket\varphi = \sum_{g\in\Supp{}\ket\varphi} \left<g\middle|\varphi\right>\ket{g}$,
	we get $0\neq \left<s \right | U \ket\varphi 
	\ = \ \sum_{g\in \Supp{}\ket\varphi} 
	\left<g \middle| \varphi \right>
	\left<s\right| U \ket{g}$,
	and so there must be at least one~$g$ for which 
	$\left<s\right| U \ket{g}\neq 0$,
	that is, $s\in \Supp{}(U\ket{g})$.
\end{proof}
\begin{proof}[of \emph{\textbf{Lemma~\ref{lem: all classical moves have an equivalent quantum move}}}]
Given~$g\in S$, write $\hat{g} := \Gamma(g,\omega(g))$
and define $\smash{\hat{S}}:=\{\Gamma(g,\omega(g))\colon g\in S\}$.
Then~$\hat{\cdot}$ gives a bijection~$S\longrightarrow \smash{\hat{S}}$,
using here that the game~$\G$ is injective.
Recall that the positions of~$\G$ are encoded using length-$L$ bitstrings,
that is, $\G\subseteq \{0,1\}^L$.
Since there is a bijection from~$S$ to~$\smash{\hat{S}}$,
the sets~$S$ and~$\smash{\hat{S}}$
have the same size.
As a consequence, $\{0,1\}^L\backslash S$
and $\{0,1\}^L\backslash \smash{\hat{S}}$
have the same number of elements too.
This means that $\hat{\cdot}\colon S\to\smash{\hat{S}}$
can be extended to a permutation  of~$\{0,1\}^L$.
Pick such a permutation $\hat{\cdot}\colon \{0,1\}^L\to\{0,1\}^L$ --- it does not matter which one.
Let~$U_\omega$ denote the unitary associated with this permutation,
	given by $U_\omega\ket{w} = \ket{\hat{w}}$.
Then, by definition, equation~\eqref{eq:cqm} holds.

Let non-finished~$\ket \varphi\in\G^Q$ with $\Supp{}\ket\varphi\subseteq S$ be given.
We claim that~$U_\omega$ is a valid move on~$\ket\varphi$.
First, $U_\omega\ket\varphi$ is a superposition of positions in~$\G^C$,
since $U_\omega$ maps each~$g\in\Supp{}\ket\varphi\subseteq S\subseteq \G$ to~$\hat{g}\in\G$.
Second, let~$g'\in \Supp{}(U_\omega\ket\varphi)$ be given;
	we must find~$g\in \G$ and $\sigma\in\Sigma$ with $g\xrightarrow\sigma g'$.
By Lemma~\ref{supp-lemma},
	there is~$g\in\Supp{}\ket\varphi$ with~$g'\in \Supp{}(U_\omega\ket{g})$.
Since~$g\in\Supp{}\ket\varphi\subseteq S$,
we have $U_\omega\ket{g}=\ket{\hat{g}}$, so $g'=\hat{g}$.
	Since $\smash{g\xrightarrow{\omega(g)}\hat{g}}\equiv g'$,
$U_\omega$ is a valid move on~$\ket\varphi$.
That~$U_\omega$
is a classical move on~$\ket\varphi$
	follows directly from~\eqref{eq:cqm}.
\end{proof}
\begin{proof}[of \emph{\textbf{Proposition~\ref{classical-bisimilar}}}]
We claim that the relation  $\sim \subseteq \G \times \G^{C\text{--}C}$ given by
	$g \sim \ket{\varphi}$ iff $\ket{\varphi}=\ket{g}$ is a bisimulation.
Indeed,
first of all, $0 \sim 0$, because $0\equiv \ket{0}$ is by definition
the starting position of $\G^{C\text{--}C}$.
	Second, we check the two numbered conditions of~Definition~\ref{def:bisimulation}:
	\begin{enumerate}
		\item 
Let $g_1 \sim \ket{\varphi_2}$ be given.
			Then $\ket{\varphi_2} = \ket{g_1}$
by definition of~$\sim$.
Now, $g_1\in\mathcal{F}$
iff $\ket{g_1}\in\mathcal{F}^Q$,
by definition of the quantum variant.
If~$g_1\in \mathcal{F}$,
then
			$\mathcal{O}(\ket{g_1})=\mathcal{O}(g_1)$,
			again by definition of the quantum variant.

		\item \begin{enumerate}
		\item If $g \xrightarrow{\sigma} g'$, then by Lemma~\ref{lem: all classical moves have an equivalent quantum move} we get that $\ket g \xrightarrow{U_\sigma} U_\sigma \ket g = \ket {g'}$, and $g'\sim \ket{g'}$.
		\item If $\ket g \xrightarrow{U} \ket {\varphi'}$, then since~$U$ must be classical, $U\ket{g}=\ket{\varphi'}\equiv \ket {g'}$ for some~$g'\in \G$.
			By definition of a valid quantum move,
				there is a $\sigma\in \Sigma^C$ with $g \xrightarrow{\sigma}g'(\sim \ket{\varphi'})$.
		\end{enumerate}
	\end{enumerate}
	Hence we have a bisimulation, $\G \sim \G^{C\text{--}C}$.
\end{proof}

\section{Quantum Byzantine consensus}\label{ap:QBC}

We model the network as a graph $G=(V,E)$, where nodes represent computing units and edges represent communication channels. We assume synchronous rounds and perfect communication channels.

Each node has a state,
an element from the set~$\mathcal{S}$,
that includes not only its internal state,
but also the registers for incoming and outgoing messages,
and its decision value from a set $\mathcal{D}$.
Every round, every node applies a local state-transition $\mathcal{S}\rightarrow \mathcal{S}$, 
after which the system routes outgoing messages according to the network topology, $G$. 
We denote the action of routing messages by 
$\mathrm{sys}\colon \mathcal{S}^n \rightarrow \mathcal{S}^n$, where $n = |V|$.

The Byzantine failure problem asks whether, in a network with up to $f$ adversarial nodes, the remaining well-behaving nodes can reach consensus. Formally we require:

\begin{enumerate}
	\item \textbf{Agreement: } all well-behaving nodes must end 
		with the same decision value; and
	\item \textbf{Validity: } if all honest nodes start with the same decision value, they must decide that value.
\end{enumerate}

A commonly considered third requirement is ``Termination''. By the finiteness requirement of quantizable games we only consider the Byzantine problem for a fixed finite number of rounds $N$. This implicitly gives us the termination requirement.

We now show how this can be cast as a combinatorial game. Let player $X$ represent the well-behaving nodes and player $O$ represent the Byzantine nodes. The game $\mathcal{G}_{BF}(f, G, N)$ is defined like so.

\begin{definition}
	Let $\mathcal{G}_{BF}(f, G, N)$ be the classical combinatorial game representing the Byzantine problem on network topology $G$ with at most $f$ failure nodes played for $N$ turns.
	
	A position is defined as a triple, $(\vec s, \mathcal{A}, B)$, the vector $\vec s \in \mathcal{S}^n$ of states of all nodes, the algorithm $\mathcal{A}$ the well-behaving nodes will follow, and the set $B \subseteq V$ of Byzantine nodes. The game starts in state $(\vec 0, NULL, \emptyset)$, where the algorithm and Byzantine have not yet been chosen. The initial state vector $\vec 0$ starts each node with empty message registers, no decision value, and empty internal state, except for a unique identifier for this node. The set of positions is formed by all positions which are reachable by the set of moves.

	The set of moves is defined as $\Sigma := \Sigma_{init} \cup \Sigma_{move}$. Where $\Sigma_{move} := \{\sigma: \mathcal{S} \rightarrow \mathcal{S}\}$ (the set of local state transitions) and $\Sigma_{init} := \{(\vec d, B)\ |\ \vec d \in \mathcal{D}^n, B \subseteq V, |B| \leq f\}$ (the set of initial problem configurations).

	The game proceeds in two phases:

	\begin{enumerate}
		\item \emph{Initialization:} Player $X$ chooses an algorithm $\mathcal{A} \in \Sigma_{move}$, yielding position $(\vec 0, \mathcal{A}, \emptyset)$.  Then player $O$ chooses an initial set-up $(\vec d, B) \in \Sigma_{init}$, yielding position $(\vec{s}_0, \mathcal{A}, B)$, where $\vec s_0$ is the result of setting the decision values in $\vec 0$ according to $\vec d$.
		\item \emph{Simulation:} At each subsequent round, player $X$'s move is fixed: first the application of the $\mathrm{sys}$ function, then the pre-chosen move $\mathcal{A}$. Where $\mathcal{A}$ is only applied to the well-behaving nodes, $V\setminus B$. Player $O$ chooses any $\sigma \in \Sigma_{move}$, which is only applied to the Byzantine nodes, $B$. 
	\end{enumerate}

	After $N$ turns, player $X$ wins iff consensus is reached i.e. agreement and validity; otherwise player $O$ wins.
\end{definition}

The important choice here is that player $X$ commits to a strategy $\mathcal{A}$ before it is known what nodes will fail or what the initial values are, effectively hiding this information from player $X$. 

\begin{corollary}
	Let $G$ be a network topology and $f$ the maximum number of Byzantine nodes. If there exists a classical algorithm guaranteeing consensus against classical Byzantine nodes within $N$ rounds, then there exists a classical algorithm guaranteeing  consensus against quantum Byzantine nodes within $N$ rounds.
\end{corollary}

\begin{proof}
A direct application of Corollary~\ref{cor: zermelo for quantum}.
\end{proof}

\begin{remark}
	One can adapt the definition to include player $X$'s choice for the algorithm used in the game definition. This would give that the algorithm that guarantees classical consensus is \emph{precisely} the classical algorithm that guarantees consensus against the quantum failures.
\end{remark}

\begin{remark}
	The only meaningful quantum variant is $\G^{C-Q}$ for this application, for two reasons. First, if player $X$ plays as a quantum player it is restricted in its choice of algorithm. It may choose any unitary in the initialization phase, however this will still be encoded as a superposition of classical moves. Player $O$ can then fix different initial values for each classical component, effectively forcing player $X$ to play classically. Second, by the definition of valid quantum moves, player $X$ still retains some freedom in its choice of move \emph{after} the initial values have been specified. This is against the spirit of the Byzantine Generals problem.

	These reasons might seem contradictory; arguing too much and too little freedom. The resolution is that the set of valid quantum moves is constrained by the allowed classical moves, but the set of valid classical moves is restricted too much. Simultaneously, within this restricted set of valid quantum moves, there is still freedom of choice, which violates the spirit of the problem. 
\end{remark}

\end{document}